\documentclass{llncs}
\usepackage{amsmath}
\usepackage{amsfonts}
\usepackage{amssymb}
\usepackage{booktabs}
\usepackage{xspace} 
\usepackage{tikz} 
\usepackage{times}
\usepackage{multirow}
\usepackage{microtype}
\usepackage{changebar}
\usepackage{pgfplots}
\usepackage{hyperref}

\usepackage{flushend}

\newcommand\iw{$i$-witness}
\newcommand\icw{$i$-cowit}
\newcommand\cu{\ensuremath{C_\_}}
\newcommand\even{\mathsf{even}}
\newcommand\odd{\mathsf{odd}}
\newcommand\pos{\mathsf{positions}}
\newcommand\unblk{\mathsf{unblocked}}
\newcommand\ubp{\mathsf{ubp}}
\newcommand\evenodd{\mathsf{evenodd}}
\newcommand\val{\mathsf{value}}
\newcommand\up{\mathsf{up}}
\newcommand\ru{\mathsf{ru}}
\newcommand\au{\mathsf{au}}
\newcommand\won{\mathsf{won}}

\newcommand\cnt{\mathsf{count}}

\begin{document}
\title{Smaller Progress Measures and Separating Automata for Parity Games}

\author{Daniele Dell'Erba, Sven Schewe}

\institute{University of Liverpool}

\sloppy

\newcommand \Aa {\mathcal{P}}
\newcommand{\seq}[1]{\langle #1 \rangle}

\newcommand{\Plays}{\mathsf{Plays}}

\maketitle

\begin{abstract}
Calude et al.\ have recently shown that parity games can be solved in quasi-polynomial time, a landmark result that has led to a number of approaches with quasi-polynomial complexity.
Jurdinski and Lasic have further improved the precise complexity of parity games, especially when the number of priorities is low (logarithmic in the number of positions).
Both of these algorithms belong to a class of game solving techniques now often called separating automata:
deterministic automata that can be used as witness automata to decide the winner in parity games up to a given number of states and colours.
We suggest a number of adjustments to the approach of Calude et al.~\cite{CJKLS17} that lead to smaller statespaces.
These include and improve over those earlier introduced by Fearnley et al.~\cite{FJKSSW19}.

We identify two of them that, together, lead to a statespace of exactly the same size Jurdzinski and Lasic's concise progress measures \cite{JL17}, which currently hold the crown as smallest statespace.
The remaining improvements, hence, lead to a further reduction in the size of the statespace, making our approach the most succinct progress measures available for parity games.
\end{abstract}

\section{Introduction}

Parity games are two-player perfect information turn-based zero-sum games of 
infinite duration played on finite directed graphs. Each vertex, that is 
labelled with an integer \emph{colour}, is assigned to one of the two players, 
\emph{even} and \emph{odd}, here referred to as \emph{he} and \emph{she}, respectively.
A play consists in an infinite sequence of player's moves around the graph, 
and the winner is determined by the \emph{parity} of the largest colour 
encountered along the play. Hence, player even (he) wins if it is an even colour, 
and player odd (she) wins otherwise.

Parity games have been extensively studied for their practical applications, 
to determine their complexity status, and to find efficient solutions.

From a practical point of view, many problems in formal verification and 
synthesis can be reformulated in terms of solving parity games. Computing 
winning strategies for these games is linear-time equivalent to solving the 
modal $\mu$-calculus model checking problem~\cite{EJS01,EL86}. They can be applied to solve the 
complementation problem for alternating automata~\cite{GTW02} or the emptiness of the 
corresponding nondeterministic tree automata~\cite{KV98}. These automata, in turn, can be 
used to solve the satisfiability and model checking problems for several 
expressive logics~\cite{CHP10,MMV10a,MMPV12,BMM13,BG04}, such as $\mu$-calculus~\cite{Wil01,SF06} and ATL*~\cite{AHK02,Sch08}.

On the complexity theoretic side, determining the winner of a parity game is 
a problem that lies in NP~$\cap$~co-NP~\cite{EJS01}, being memoryless determined~\cite{Mos91,EJ91,Mar75}, but 
it has been even proved to belong to UP~$\cap$~co-UP~\cite{Jur98}, 
and later to be solvable in quasi-polynomial time~\cite{CJKLS17}. However, determining their 
exact complexity is still an open problem.

The existing algorithms for solving parity game can be divided into two classes.
The first one collects approaches that solve the game by creating a winning 
strategy for one of the two players on the entire game. This can be done either 
employing a value iteration over progress measures~\cite{Jur00} or iteratively improving the 
current strategy~\cite{VJ00,Fea10a,Fri13}. To the second class, instead, belong approaches that 
decompose the solution of a game into the analysis of its subgames in a 
divide-et-empera concept. To do so, these approaches partition the game into a set of positions that satisfy the required properties. The name of the sets depend on the properties: attraction set~\cite{Zie98,Sch07}, region~\cite{BDM18,BDM18b,BDM16b}, tangle~\cite{Dij18a}, and justification~\cite{LBD20}.

Many algorithms from both the classes have been refined to achieve a 
quasi-polynomial upper bound since the contribution of Calude et al.\
\cite{CJKLS17}. This seminal paper works as a value iteration algorithm 
with compact measures for which a poly-logarithmic size witness is sufficient, 
rather than storing the entire history of the play. The same approach has been refined improving the complexity result~\cite{FJSSW17,JL17}, while the same complexity has been achieved by a number of different approaches such as the register-index algorithm~\cite{Leh18} and the bounded version of the recursive algorithm~\cite{Par18}. Interestingly, all the known quasi-polynomial algorithms have been proved to be derived by the separation 
approach that also provide a lower bound for these techniques~\cite{CDFJLP18}.

\subsection{Contribution}
We adjust the definitions of \emph{witnesses}, the data structure first used by Calude et al.~\cite{CJKLS17} and the way they are updated in a number of ways.

The most clear-cut improvement is the increased succinctness of the resulting structures: integrating them with the improvements suggested in Fearnley et al.~\cite{FJKSSW19}, we achieve a number of improvements with different power to improve the succinctness of the data structure.
The most powerful of these improvements is the restriction of the occurrences of odd colours within the witnesses to at most one.
Together with the small improvement from \cite{FJKSSW19} that, when the maximal colour is odd, we can just reset the witness to its initial value instead of recording this value, we obtain a statespace of quite different strucuture to, but the same size as, the currently smallest statespace from Jurdzinski and Lasic \cite{JL17}.

On top of this improvement, we offer small additional improvements, incl.\ the improvement from \cite{FJKSSW19} that the odd colours can be skipped for the least significant position of a witness, and the new improvements that make a more careful use bounds on the length of `even chains' (usually the number of positions with even colour) and the exclusion of the least colour when this colour is odd.
Depending on the bound for the length of an even chain, this translates to a further improvement between a factor from just under two to just under four.

The second improvement is a re-definition of the semantics of witnesses, moving from the classic witnesses to \emph{colour witnesses}, where all positions with the came colour in a witness refer to one chain, instead of referring to many.

This accelerated convergence, though the acceleration is muted where it is used for value iteration.

\subsection{Outline}
We discuss a variation of the algorithm of Calude et al.~\cite{CJKLS17}, partly in the original version and partly in the variation suggested by Fearnley et al.~\cite{FJKSSW19}, to extend this approach to value iteration.

After the general preliminaries, we therefore recap this approach in Section \ref{sec:classic}, using a mild variation of the witness from \cite{CJKLS17} for a basic update rule $\up'$ that updates a witness $\mathbf b$ when reading a state with colour $v$ to a witness $\mathbf b' = \up'(\mathbf b,v)$, and an antagonistic update rule $\au'$ that updates a witness $\mathbf b$ when reading a state with colour $v$ to the witness $\min_{\mathbf c {\sqsupseteq'} b} \up'(\mathbf c,v)$.

We then amend those rules in two steps.
The first step (Section \ref{sec:concise}) reduces the statespace, but otherwise retains the classic lines of \cite{FJKSSW19}.
It is a simple extension that carries the easiest to spot (and sell) improvement of this work: the reductions from the statespace of the witnesses used, leading to more concise witnesses.

The backbone of the statespace reduction is to simply restricts the number of times an odd colour occurs in a witness to at most once.
Where the maximal colour is even, this change alone leads to a perfect match in size with the statespace of \cite{JL17}, which is currently the smallest.
This perfect match is a bit surprising, as the structure of the statespace is very different.

The remaining changes extend this to the case where the maximal colour is odd and collects some further minor reductions that roughly lead to a spatespace reduction that is usually in a range between $2$ and $4$, where the advantage is strongest when the number of states with an even number of colours is a power of $2$.

Subsequently, we re-interpret the semantics of a witness in Section \ref{sec:colour}.
This change in semantics does not change the statespace, but it allows for updating the witnesses faster.
Faster updating is mainly improving the basic update rule, but to some extend also the antagonistic update rule, leading to faster convergence in both cases.

We then turn to an estimation of the new statespace in Section \ref{sec:complex}.
For this, we proceed in a number of steps.
We first look at the two classic statespaces, considering the previously most concise one those of Calude et al.'s original QP algorithm \cite{CJKLS17}.

We then turn to using those improvements to \cite{CJKLS17} that lead to a statespace of size equal to that of \cite{JL17}.
This is done in order to be able to show that the two statespaces are of precisely the same size, but also to have a clear understanding which improvements remain beyond this, and to focus on how they influence the statespace.

We then exemplify how the three statespaces compare in size in Section \ref{sec:compare}.

\section{Preliminaries}
\label{sec:prelim}
Parity games are turn-based zero-sum games played
between two players---even and odd, referred to as he and she, respectively---over finite graphs.
A parity game $\Aa$ is a tuple $(V_e, V_o, E, C, \phi)$, where
\begin{itemize}
    \item $(V = V_e \cup V_o, E)$ is a finite directed graph, where the set $V$ of vertices is partitioned into a set $V_e$ of vertices controlled by player \emph{even}
    and a set $V_o$ of vertices controlled by player \emph{odd}, and where $E \subseteq V \times V$ is the set of edges;
    \item $C\subseteq \mathbb N = \{1,2,3, \ldots\}$ is a finite consecutive set of colours, such that $C = \{1,2, \ldots, \max\{C\}$ or $C = \{2,3, \ldots, \max\{C\}$ holds; and
    \item $\phi: V \to C$ is the colouring functions that maps each vertex to a colour.
\end{itemize}

We define $C^- = C \smallsetminus \big\{\max\{C\}\big\}$ if the highest colour $\max\{C\}$ is odd, and $C^- = C$ if the highest  colour $\max\{C\}$ is even, and require that every vertex has at least one outgoing edge. 

Intuitively, a parity game $\Aa$ is played between the two players
by moving a token along the edges of the directed graph $(V,E)$. 
A play of such a game starts at some initial vertex 
$v_0 \in V$ where the token is placed at the beginning.
The player controlling this vertex then chooses a successor vertex
$v_1$ such that $(v_0, v_1) \in E$, and the token is moved to this successor vertex. 
In the next turn the player controlling the vertex $v_1$ make his choice by picking a successor
vertex $v_2$ where to move the token, such that $(v_1, v_2) \in E$, and so on.
In this manner both players move the token over the arena and thus form an infinite play of the game.

Formally, a play of a game $\Aa$ is an infinite sequence of vertices
$\seq{v_0, v_1, \ldots} \in V^\omega$ such that, for all $i \geq 0$, we have that
$(v_i, v_{i+1}) \in E$.  
We denote as $\Plays_\Aa(v)$ the set of plays of the game $\Aa$ that origins in a vertex
$v \in V$ and as $\Plays_\Aa$ the set of all plays of the game. 
We omit the subscript when the arena is clear from the context. 
The colour mapping $\phi: V \to C$ can be extended from vertices to plays by defining
the mapping $\phi: \Plays \to C^\omega$ as  
$\seq{v_0, v_1, \ldots} \mapsto \seq{\phi(v_0), \phi(v_1), \ldots}$.

A play $\seq{v_0, v_1, \ldots}$ is won by player \emph{even} if $\limsup_{i \rightarrow \infty}\phi(v_i)$ is even, and by player \emph{odd} otherwise.

A \emph{prefix} of a play (or play prefix) is a non-empty initial sequence $\seq{v_0, v_1, \ldots, v_m}$ of a play $\seq{v_0, v_1, \ldots}$.

For a play $\rho=\seq{v_0, v_1, \ldots}$ or play prefix $\rho=\seq{v_0, v_1, \ldots,v_n}$, an \emph{even chain} of length $\ell$ is a sequence of positions $p_1 < p_2 <p_3 <
\ldots < p_\ell$ (with $0\leq p_1$ and, for plays prefixes, $p_\ell \leq n$) in $\rho$ that has the
following properties:
\begin{itemize}
\item for all $j \in \{1, \ldots, \ell\}$, we have that $\phi(v_{p_j})$ is even, and
\item for all $j \in \{1, \ldots, \ell-1\}$ the colours in the subsequence defined by $p_{j}$ and
$p_{j+1}$ are less than or equal to  $\phi(p_j)$ or $\phi(p_{j+1})$.
More
formally, we have that all colours
$\phi(v_{p_j}),\phi(v_{(p_j)+1}),\ldots,\phi(v_{p_{(j+1)}})$ are less than or
equal to $\max\big\{\phi(v_{p_j}),\phi(v_{p_{j+1}})\big\}$.
\end{itemize}

A strategy for player \emph{even} is a function $\sigma: V^*V_e \rightarrow V$ such that $\big(v,\sigma(\rho,v)\big)\in E$ for all $\rho \in V^*$ and $v \in V_e$.
If a strategy $\sigma$ only depends on the last state, then is called memoryless ($\sigma(\rho,v) = \sigma(\rho',v)$ for all $\rho,\rho' \in V^*$ and $v \in V_e$).
A play $\seq{v_0, v_1, \ldots}$ is consistent with $\sigma$ if, for every initial sequence $\rho_n = v_0, v_1, \ldots,v_n$ of the play that ends in a state of player \emph{even} ($v_n \in V_e$), $\sigma(\rho_n)=v_{n+1}$ holds.
Player \emph{even} wins the game starting at $v_0$ if he has a strategy $\sigma$ such that either all plays $\seq{v_0, v_1, \ldots}$ consistent with $\sigma$ satisfy $\limsup_{i \rightarrow \infty}\phi(v_i)$ (i.e.\ the highest colour that occurs infinitely often in the play) is even or all plays $\seq{v_0, v_1, \ldots}$ consistent with $\sigma$ contain a loop $v_i, v_{i+1}, \ldots, v_{i+k}$, that satisfies $v_i = v_{i+k}$ and that $\max\{\phi(v_i), \ldots, \phi(v_{i+j})$ is even. In both cases $\sigma$ might be memoryless. We use different criteria in the technical part, choosing the one that is most convenient.

A \emph{separating automaton} \cite{automataToolbox} for parity games with a set of colours $C$ and a bounded number of states, or a bounded number of states with even colour,
is a deterministic reachability automaton $\mathcal A = (Q , C , q_0 , \delta, \won)$, where
\begin{itemize}
    \item $Q$ is the set of states, with $q_0,\won \in Q$, $q_0$ is the initial state and $\won$ is the target state (and sink), and
    \item $\delta: Q \times C \rightarrow Q$ is the transition function (with $\delta(\won,v)=\won$ for all $v\in C$),
\end{itemize}
such that, for all parity games with colours $C$ and that have no more states (of even colour) than the given bound, there are
\begin{itemize}
    \item if $q \in V$ is a winning state, then there is a positional strategy $\sigma$ for player even such that the run of $\mathcal A$ on all plays consistent $\sigma$ are accepted by $\mathcal A$, i.e.\ reach the target state $\won$; and
    \item if $q \in V$ is a winning state, then there is a positional strategy $\sigma$ for player odd such that the run of $\mathcal A$ on all plays in $\Plays(q)$
    consistent $\sigma$ are accepted by $\mathcal A$, i.e.\ does not reach the target state $\won$.
\end{itemize}

\section{Classic Witnesses}
\label{sec:classic}
We adjust the approach from \cite{CJKLS17} and \cite{FJKSSW19}, and this section is predominantly taking the representation from \cite{FJKSSW19}.
It does, however, change some details in the definitions of $i$-witnesses that end in an odd priority and the definition of the value of a witness slightly to suite the rest of the paper better.
Where the proofs are affected, they are adjusted and given, but the proofs are mostly unaffected by these minor details.

\subsection{Classic Forward Witness}
We start with describing the \emph{old} witness without making its semantics formal (as we do not need it in this paper), and will turn to the new \emph{concise witness} (Section \ref{sec:concise}) and the \emph{colour witness} (Section \ref{sec:colour}) afterwards.

\paragraph{\bf $i$-Witnesses}
Let $\rho = v_1, v_2, \dots, v_m$ be a prefix of 
a play of the parity game. An \emph{even \iw} is a sequence of (not necessarily consecutive) positions of $\rho$
\begin{equation*}
p_1, p_2, p_3, \dots, p_{2^i}, 
\end{equation*}
of length exactly $2^i$,
and an \emph{odd \iw}
is a sequence of (not necessarily consecutive) positions
of $\rho$ 
\begin{equation*}
p_0, p_1, p_2,\dots, p_{2^i}
\end{equation*}
of length exactly $2^i+1$,
that satisfy the following properties:
\begin{itemize}
\item \textbf{Position:} each $p_j$ specifies a position in the play $\rho$,
so each $p_j$ is an integer that satisfies $1 \le p_j \le m$.
\item \textbf{Order:} the positions are ordered. So we have $p_j < p_{j+1}$ for
all $j < 2^i$.

\item {\bf Evenness:} 
all positions but the final one are even. Formally, for all $j < 2^i$ the
colour $\phi(v_{p_{j}})$ of the vertex in position $p_{j}$ is even.

For position $p_{2^i}$, its colour $\phi(v_{p_{2^i}})$ is even for an even \iw, and odd for an odd \iw.

Note that this entails that an \iw\ contains an even chain of length $2^i$.

\item {\bf Inner domination:}  
the colour of every vertex between $p_j$ and $p_{j+1}$ is dominated by the colour
of $p_j$ \emph{or} the colour of $p_{j+1}$.
Formally, 
for all $j < 2^i$, the colour of every vertex in the subsequence
$v_{p_{j}},v_{(p_{j})+1},\ldots,v_{p_{(j+1)}}$ is less than or equal to 
$\max\big\{\phi(v_{p_{j}}),\phi(v_{p_{j+1}})\big\}$.

\item {\bf Outer domination:} 
the colour of $p_{2^i}$ is greater than or equal to the colour of every vertex
that appears after 
$p_{2^i}$ in $\rho$.
Formally, for all $k$ in the range $p_{2^i} < k \le m$, we
have that $\phi(v_{k}) \le \phi(v_{p_{2^i}})$.
\end{itemize}

\paragraph{\bf Witnesses}
We define $\cu = C^- \cup \{ \_ \}$ as the set of colours plus with the $\_$
symbol.
A \emph{witness} is a sequence\footnote{While $k$ can be viewed as "big enough" or as "of arbitrary size" for the definition, we will later see that a length $k+1$, with $k = \lfloor \log_2(e)\rfloor$, where $e$ is the number of vertices with an even colour, or any other sufficient criterion for the maximal length of an even chain, is sufficient.} 
\begin{equation*}
b_k, b_{k-1}, \dots, b_1, b_0,
\end{equation*}
such that each element $b_i \in \cu$, and
that satisfies the following properties:
\begin{itemize}
\item \textbf{Witnessing:} there exists a family of $i$-witnesses, one for each element $b_i$ with $b_i \ne \_$.
We refer to such an $i$-witness in the run $\rho$. We will refer to this witness as
\begin{equation*}
p_{i, 1}, \; p_{i, 2}, \; \dots, \; p_{i, 2^i}
\end{equation*}
for even witnesses and
\begin{equation*}
p_{i, 0}, \; p_{i, 1}, \; \dots, \; p_{i, 2^i}
\end{equation*}
for odd witnesses.
\item \textbf{Dominating colour:}
For each $b_i \ne \_$, we have that $b_i = \phi(v_{p_{i, 2^i}})$. That is,
$b_i$ is the outer domination colour of the $i$-witness.
\item \textbf{Ordered sequences:} The $i$-witness associated with $b_i$ starts after a
$j$-witness associated with $b_{j}$ whenever $i < j$. Formally, for all $i$
and $j$ with $i < j$, if $b_i
\ne \_$ and $b_{j} \ne \_$, then $p_{j, 2^j} < p_{i, 1}$ when the \iw\ is even, and $p_{j, 2^j} < p_{i, 0}$ otherwise.
\end{itemize}

For a little bit of extra conciseness, we also require that $b_0$ is either even or $\_$.

Note that the witness does not store the $i$-witnesses associated with each position~$b_i$.
However, the sequence is a witness only if the corresponding $i$-witnesses \emph{exist}.
Moreover, the colours in a witness are monotonically increasing for growing indices (and thus increase from right to left), since each colour $b_j$ (weakly) dominates all colours that appear
afterwards in~$\rho$ as a consequence of the dominating colour property and the ordered sequences property.

\paragraph{\bf Forward and backward witnesses.}
The \emph{forward} witnesses described so far were introduced in~\cite{CJKLS17}, while we now describe the \emph{backward} witnesses and an ordering over them that have been introduced in~\cite{FJKSSW19}.
For each play prefix $\rho = v_1, v_2,
\dots, v_m$, we define a reverse play $\overleftarrow{\rho} = v_m, v_{m-1},
\dots, v_1$; a backward witness is a witness for $\overleftarrow{\rho}$, or for a prefix of it.

\paragraph{\bf Order on witnesses.}
The order $\succeq$ we have mentioned earlier orders the set $\cu$ such that even numbers are better than odd numbers, higher even numbers are better than smaller even numbers, smaller odd numbers are better than
higher odd numbers, and every number is better than $\_$.
Formally, $a \succeq b$ if $b = \_\,$; or $a$ is even and $b$ is either odd or $a \geq b$; or $a \leq b$ and they are both odd.

Using $\succeq$, we define an order ${\sqsupseteq'}$ over witnesses that compares two witnesses lexicographically, where the most significant element is $b_k$ and the least significant element is $b_0$. Each element is compared using the order $\succeq$.
The biggest witness has a special value $\won$; i.e., $\won \sqsupseteq' \mathbf b$ holds for all witnesses $\mathbf b$.

\paragraph{\bf The value of a witness.}

For each witness $\mathbf b = b_k, b_{k-1}, \ldots, b_0$, we define the following functions:
\begin{itemize}
\item \textbf{Even positions:} $\even(\mathbf b) = \{i \in \mathbb N_0 \mid b_i \mbox{ is an even number} \}$,
\item  \textbf{Relevant \iw es:} $\evenodd(\mathbf b) = \even(\mathbf b)$ if $\mathbf b$ does not contain an odd number,

otherwise, $\evenodd(\mathbf b) = \{i \in \even(\mathbf b) \mid i > o\} \cup \{o\}$,

with $o = \max\{i \in \mathbb N \mid b_i$ is odd$\}$, and 

\item  \textbf{Value of witness:} $\val(\mathbf b) = \sum\limits_{i\in \evenodd(\mathbf b)} 2^i$.
\end{itemize}

\noindent\textbf{Remark.}
The value function from \cite{FJKSSW19} is different in that is uses $\sum\limits_{i\in \even(\mathbf b)} 2^i$.
We will discuss the impact that this difference has on the statespace at the end of Section \ref{sec:complex}.
\medskip

We can show that the value of $\mathbf{b}$ corresponds to the
length of an even chain in $\rho$ that is witnessed by $\mathbf{b}$.

\begin{lemma}\cite{FJKSSW19}
If\ $\mathbf{b}$ is a (forward or backward) witness of $\rho$, then there 
is an even chain of length $\val(\mathbf b)$ in $\rho$.
\end{lemma}

If we count the number of vertices with even colours in the game as $e = |\{v \in V \; : \; \phi(v) \text{ is even }\}|$, then we can observe that in case we have an even chain longer than $e$ then $\rho$ contains a cycle,
as there is a vertex with even colour visited twice in this even chain.
Moreover, the cycle is winning player \emph{even}, since the largest priority of its vertices must be even. As a consequence, if player \emph{even} can force a play that has a witness whose value is strictly
greater than $e$, he wins the game.

\begin{lemma}\cite{FJKSSW19}
\label{lem:correct}
If, from an initial state $v_0$, player \emph{even} can force the game to run
through a sequence $\rho$, such that $\rho$ has a (forward or backward)
witness $\mathbf b$ such that $\val(\mathbf b)$ is greater than the number of vertices
with even colour, then player \emph{even} wins the parity game starting at $v_0$.
\end{lemma}

For this reason, we only need witnesses with value $\leq e$.
If an update would produce a witness of value $>e$, then the resulting witness must contain a winning cycle.

Thus, the set of classic witnesses is $\mathbb W = \{\mathbf b \mid \mathbf b \mbox{ is a witness with } \val(b) \leq e\} \cup \{\won\}$.

\subsection{Updating Witnesses}

Forward witnesses can be constructed incrementally by processing
the play one vertex at a time. The following lemmas assume that
we have a play $\rho = v_0, v_1, \dots, v_{m}$, and a new vertex $v_{m+1}$ that
we are going to append to $\rho$ in order to create $\rho'$. The value $d =
\phi(v_{m+1})$ denotes the colour of the new vertex $v_{m+1}$. We will suppose that
$\mathbf b = b_k, b_{k-1}, \ldots, b_1,b_0$ is a witness for $\rho$, and we will
construct a witness $\mathbf c = c_k, c_{k-1}, \ldots, c_1,c_0$ for $\rho'$.

We present three lemmas that allow us to perform this task.

\begin{lemma}\cite{FJKSSW19}
\label{lem:up.overflow}
Suppose that $d$ is even, there exists an index $j$ such that:
\begin{itemize}
    \item $b_i$ is even for all $i < j$,
    \item $b_j$ is odd or equal to $\_$, and
    \item $b_i \ge d$ or equal to $\_$ for all $i > j$.
\end{itemize}
If we set $c_i = b_i$ for
all $i > j$, $c_j = d$, and $c_i = \_$ for all $i < j$, then $\mathbf{c}$ is a
witness for $\rho'$.
\end{lemma}

Note that we returned to the original definition from Calude et al.~\cite{CJKLS17} by restricting the `overflow rule' from Lemma \ref{lem:up.overflow} to even numbers, whereas the witnesses from \cite{FJKSSW19} also allowed this operation
to be performed in the case where $d$ is odd.
The reason for this is that it reduces the statespace:
while this reduction is insignificant in most cases, it is quite substantial if $e=2^p-1$ for some power $p \in \mathbb N$, as it leads to an increase in the length of the witness.
As this statespace reduction is a core target of this paper, we opted to be precise here.

Note that the following lemmas (and their proofs) are essentially independent of this.

\begin{lemma}\cite{FJKSSW19}
\label{lem:up.local}
Suppose that $d \in C^-$ and there exists an index $j$ such that:
\begin{itemize}
    \item $d > b_j \neq \_$ and
    \item $b_i \geq d$ or equal to $\_$ for all $i>j$.
\end{itemize}
Then setting $c_i = b_i$ for all $i > j$, setting $c_j = d$ if $j\neq 0$ (and $c_j=\_$ if $j=0$), and setting $c_i = \_$ for all $i < j$ yields a witness for $\rho'$.
\end{lemma}

There is a tiny difference in the proof of this lemma in that we require the length of the $j$-witness to be $2^j+1$ when $c_j$ is set to $d$. 
But either $b_j$ was odd before, in which case replacing the last index of the $j$-witness by $m+1$ still produces a witness of length $2^j+1$, or it was even, and in that case we can instead append $m+1$ to the old $j$-witness.

\begin{lemma}\cite{FJKSSW19}
\label{lem:up.stale}
Suppose that $d \in C^-$ is odd and, for all $j \leq k$, either $b_j = \_$ or $b_j \geq d$. If we
set $c_i = b_i$ for all $i \leq k$ (i.e.\ if we set $\mathbf c = \mathbf b$), then $\mathbf c$ is a witness for $\rho'$.
\end{lemma}

When we want to update a witness with the raw udate rule upon scanning another state $v_{m+1}$ with colour $d = \phi(v_{m+1})$, we select the according lemma if $d \in C^-$.
Otherwise, i.e.\ when $d = \max\{C\}$ and odd, we re-set the witness to $\_,\ldots,\_$ (which is a witness for every play prefix).

For a given witness $\mathbf b$ and a vertex $v_{m+1}$, we denote with
\begin{itemize}
 \item \textbf{Raw update:} $\ru'(\mathbf b,d)$ the raw update of the witness to $\mathbf c$, as
obtained by the update rules described above.
 \item \textbf{Update:} $\up'(\mathbf b,d)$ is either $\ru'(\mathbf b,d)$ if
$\val\big(\ru(\mathbf b,d)\big) \leq e$ (where $e$ is the number of
vertices with even colour), or $\up'(\mathbf b,d) = \won$ otherwise.

In particular $\up'(\won,d) = \won$ holds for all $d \in C$.
\item \textbf{Antagonistic update:} $\au'(\mathbf b,d) = {\min}_{\sqsubseteq'}\big\{\up'(\mathbf c,d) \mid \mathbf b {\sqsubseteq'} \mathbf c \in \mathbb W \big\}$.
\end{itemize}

\paragraph{\bf Basic and Antagonistic Update Game}
With these update rules, we define a forward and a backward basic update game played between the two players \emph{he} and \emph{she}.
In this game, they produce a play of the game as usual: if the pebble is on a \emph{his} position, then he selects a successor, and if the pebble is on a \emph{her} position, then she selects a successor.

Player \emph{even} can stop any time he likes and evaluate the game using $\mathbf b_0 = \_,\ldots,\_$ as a starting point and the update rule $\mathbf b_{i+1} = \up'(\mathbf b_i,v_i)$ (in the basic update game) and $\mathbf b_{i+1} = \au'(\mathbf b_i,v_i)$ (in the antagonistic update game), respectively.

For a forward game, he would process the partial play $\rho^+ = v_0, v_1, v_2, \ldots, v_n$ from left to right, and for the backward game he would process the partial play $\rho^- = v_n, v_{n-1}, \ldots, v_0$.
In both cases, he has won if, and only if, $\mathbf b_{n+1} = \won$.

\begin{theorem}\cite{FJKSSW19}
\label{theo:basic}
If, and only if, player \emph{even} has a strategy to win the classic forward resp.\ backward basic resp.\ antagonistic update game, then he has a strategy to win the parity game.
\end{theorem}

This can be formulated in a way that, for a given set $C$ of colours and $e$ states with even colour, the deterministic reachability automaton with states $\mathbb W$, initial state $\_,\_, \ldots,\_$, update rules $\up'$ (or $\au'$), and reachability goal to reach $\won$ is a separating automaton.

\begin{corollary}
For a parity game with $e$ states and $k=\lfloor \log_2(e) \rfloor$ and $\mathbb W$ the space for witnesses of value $\leq e$, length $k+1$ and colours $C$, both
$\mathcal U=(\mathbb W;C;\_,\ldots,\_;\up';\won)$ and $\mathcal A=(\mathbb W;C;\_,\ldots,\_;\au';\won)$ are separating automata.
\qed
\end{corollary}

The advantages of the antagonistic update rule is that it is monotone: $\mathbf{b}\sqsubseteq'\mathbf{c} \rightarrow \au'(\mathbf{b},d) \sqsubseteq \au'(\mathbf{c},d)$.
This allows for using $\au'$ in a value iteration algorithm \cite{FJKSSW19}.

\section{Concise Witness}
\label{sec:concise}

As in this article we suggest a change in the semantics of the witness, which reduces the statespace of witnesses to $\mathbb C \subset \mathbb W$, in the following section
we will improve the update rule.

The main theoretical advancement is the smaller statespace, as it directly translates into improved bounds, slightly outperforming the currently leading QP algorithm in this regard.

For this, we define a truncation operator $$\downarrow_1 \colon \mathbb W \rightarrow \mathbb C$$
that, for every odd colour $o \in C^-$, leaves only the leftmost occurrences of $o$ in a witness and replaces all other occurrences of $o$ in $\mathbf b$ by $\_$.

For example, $\downarrow_1 \_,7,\_,7,5,4,\_,3,3,\_,2 = \_,7,\_,\_,5,4,\_,3,\_,\_,2$, and $\downarrow_1 3,3,2 = 3,\_,2$.
We also have $\downarrow_1 \won = \won$.

Note that the definition of $\downarrow_1$ entails $$\even(\downarrow_1 \mathbf b) = \even(\mathbf b)\ ,$$
as well as $$\val(\downarrow_1 \mathbf b) = \val(\mathbf b)\ .$$

Building on the definition of $\downarrow_1$, we continue with the following definitions.

\begin{itemize}
\item $\mathbb C = \{\downarrow_1 \mathbf b \mid \mathbf b \in \mathbb W\}$,
\item \textbf{Raw update:} $\ru(\mathbf b,d) = \downarrow_1 \ru'(\mathbf b,d)$,
\item \textbf{Update:} $\up(\mathbf b,v) = \downarrow_1 \up'(\mathbf b,d)$,
\item \textbf{Order over witnesses:} for all $\mathbf b, \mathbf c \in \mathbb C$, $\mathbf b \sqsubseteq \mathbf c$ if, and only if, $\mathbf b \sqsubseteq' \mathbf c$  (i.e.\ $\sqsubseteq$ is simply a restriction of $\sqsubseteq'$ from $\mathbb W$ to $\mathbb C$), and

\item \textbf{Antagonistic update:} $\au(\mathbf b,d) = {\min}_\sqsubseteq\big\{\up(\mathbf c,d) \mid \mathbf b \sqsupseteq \mathbf c \in \mathbb C \big\}$

\end{itemize}

\begin{lemma}
\label{lem:ru}
If $\mathbf b = \downarrow_1 \mathbf c$, then $\ru(\mathbf b,d) = \downarrow_1 \ru'(\mathbf c,d)$.
\end{lemma}

\begin{proof}
We look at the effect the different update rules have on $\mathbf b$ and $\mathbf c$.
Lemma \ref{lem:up.overflow} would (for the same $j$) change the tail (starting with the $j$-witness) of $\mathbf c$ and $\mathbf b$ in the same way to $d,\_,\ldots,\_$, and they either both do or do not satisfy the prerequisites for its application.
Thus, the $\downarrow_1$ operator would remove exactly those positions $>j$ from $\ru'(\mathbf c,d)$ that it removed from $\mathbf c$.

When Lemma \ref{lem:up.local} applies, then it does so for the same index $j$, and it simply overwrites the tail starting there with  $d,\_,\ldots,\_$ (or with $\_$ if $j=0$).
As all higher positions are unchanged and greater or equal to $d$, $\mathbf b = \downarrow_1 \mathbf c$ implies $\ru(\mathbf b,d) = \downarrow_1 \ru'(\mathbf c,d)$.

When the conditions of Lemma \ref{lem:up.stale} apply either for both, $\mathbf b$ and $\mathbf c$, or for neither of them, then we note that Lemma \ref{lem:up.stale} does not change the witness.

Finally, if $d = \max\{C\}$ and odd, then $\ru'(\mathbf b,d) = \ru'(\mathbf c,d) = \_,\ldots,\_$, which implies $\ru(\mathbf b,d) = \downarrow_1 \ru'(\mathbf b,d) = \downarrow_1 \_,\ldots,\_ = \downarrow_1 \ru'(\mathbf c,d)$.
\qed
\end{proof}

\begin{corollary}
\label{cor:up}
If $\mathbf b = \downarrow_1 \mathbf c$, then $\up(\mathbf b,v) = \downarrow_1 \up'(\mathbf c,v)$ and $\au(\mathbf b,v) \sqsupseteq \downarrow_1 \au'(\mathbf c,v)$.
\qed
\end{corollary}

\begin{theorem}
\label{theo:concise.basic}
If, and only if, player \emph{even} has a strategy to win the concise forward resp.\ backward basic antagonistic update game, then he has a strategy to win the parity game.
\end{theorem}

\begin{proof}
This follows from Theorem \ref{theo:basic}: because the same runs are winning when using $\up$ and $\up'$ due to Corollary \ref{cor:up}, the same player wins the classic and the concise basic update game.
\qed
\end{proof}

\begin{theorem}
\label{theo:concise.antagonistic}
If, and only if, player \emph{even} has a strategy to win the concise forward resp.\ backward antagonistic update game, then he has a strategy to win the parity game.
\end{theorem}

\begin{proof}
For the 'if' case, we observe that Corollary \ref{cor:up} implies with the monotonicity of $\au$ that, when even wins the classic antagonistic update game, he also wins the concise antagonistic update game (with the same strategy). Together with Theorem \ref{theo:basic}, this provides the `if' case. %theo:antagonistic

For the `only if' case, we observe that the monotonicity of $\au$ entails that, when odd wins the basic concise update game, then she wins the antagonistic update game (with the same strategy).
Together with Theorem \ref{theo:concise.basic}, this provides the `only if' case.
\qed
\end{proof}

\begin{corollary}
For a parity game with $e$ states of even colour and $k=\lfloor \log_2(e) \rfloor$ and $\mathbb W$ the space for witnesses of value $\leq e$, length $k+1$ and colours $C$,
both $\mathcal U=(\mathbb C;C;\_,\ldots,\_;\up;\won)$ and $\mathcal A=(\mathbb C;C;\_,\ldots,\_;\au;\won)$ are separating automata.\hspace*{-12pt}
\qed
\end{corollary}

To give an intuition to their states, for $\mathcal U$ being in a state $\mathbf b \neq \won$ means that $\mathbf b$ is a witness for the play prefix, while $\won$ means that the play prefix contains an even chain of length $>e$, and thus an even cycle.

For $\mathcal A$ being in a state $\mathbf b$ means that there is a state $\mathbf c \sqsupseteq \mathbf b$ with this property.
\medskip

As a final remark, in the rare cases where even colours are scarce, their appearance in $\mathbb C$ can also be restricted: if only $e_\#$ states have an even colour $e$, then the number of occurrences of $e$ in a concise witness can be capped to $e_\#$, too, as more occurrences of $e$ without an intermediate occurrence of a higher colour would imply that an accepting cycle is in the word.

However, for this to reduce the statespace, $e_\# \leq \lfloor \log_2(e) \rfloor$ is required, and the closer it comes to $\lfloor \log_2(e) \rfloor$, the lesser is the saving. In particular, for $e_\# = \lfloor \log_2(e) \rfloor$, we would just safe a single state.

\section{Colour Witnesses}
\label{sec:colour}
In this section, we use the same data structure as before---the concise witnesses from the previous section---but adjusting its semantics.

We introduce two changes to the semantics of witnesses that accelerate the speed in which updates can be made.
We discuss them for concise witnesses $\mathbb C$.

Before formalising how we make our witnesses more flexible and how we use this to re-define the raw update function (and, through this, the update function and the antagonistic update), we describe on a number of examples how we change the semantics of witnesses.

\subsection{Motivating examples}
\paragraph{\bf Merging witnesses:}
If we consider the classic witness $\mathbf b = 4,\_,4,\_$, it referred to two $i$-witnesses that each end on a state with colour $4$, one of length $8$ and one of length $2$.

We will instead view this as a \emph{single} colour witness for colour $4$, which then refers to a single even chain of length \emph{at least} $10$.

As a consequence, when passing by a state with colour $6$, we can now update the witness to $6,\_,6,6$, as this would require a single even chain of at least length $11$ that ends in a $6$.

\paragraph{\bf Shifting witnesses:}
If we consider the witness $\mathbf b = 4,2,\_,\_$, it referred to two $i$-witnesses, where the first has length $8$ and ends on colour $4$, while the second has length $4$ and ends on colour $2$.

We will allow to make the latter sequence shorter, so long as the former sequence is extended accordingly.
For example, when the sequence that ends in $4$ has length $10$, then it would suffice if the sequence that ends in $2$ has length $2$.

Likewise, for $\mathbf b = 6,\_,4,2,2$, it would be allowed that the length of the even chain that ends in $6$ is $18$, the subsequent sequence that ends in $4$ is $3$, and the length of the sequence that ends in $2$ is $2$.
If the length of the sequences ending in $6$, $4$, and $2$ are $\ell_6$, $\ell_4$, and $\ell_2$, respectively, the constraints would be $\ell_6 \geq 16$, $\ell_6+\ell_4 \geq 20$, and $\ell_6+\ell_4+\ell_2 \geq 23$.

When passing by a state with colour $8$, we can now update the witness to $8,8,\_,\_,\_$, as this would require a single sequence of at least length $24$ that ends in an $8$.

\paragraph{\bf Blocked shifting:}
This shifting cannot be done through an odd colour: for $\mathbf b = 4,3,2,2$, the requirement for the rightmost sequence would be to be of length at least three and to end in a $2$.
It is, however, possible to shift some of the required length of the sequence that ends in $3$ to the sequence that ends in $4$:
if the length of the sequences ending in $4$ and $3$ are $\ell_4$ and $\ell_3$, respectively, the constraints would be $\ell_4 \geq 8$ and $\ell_4+\ell_3 \geq 13$. (Recall the odd witnesses need to be one position longer to contain an even chain of the same length.)  
Thus, reading a $6$ would lead to the witness $6,6,\_,6$.

\subsection{Colour Witness}
The biggest change is that an $i$-colour witness (\icw) refers to the \emph{colour} $i$, rather than to the position $b_i$ in the witness.
Consequently, we do not have a fixed length of an $i$-colour witness, and refer to the length of the witness for each colour $i$ that appears in a witness as $\ell_i$.

As before, we focus in our description on forward witnesses, with backward witnesses being defined accordingly.

\paragraph{\bf $i$-Colour witness (\icw) with value $\ell_i$}
Let $\rho = v_1, v_2, \dots, v_m$ be a prefix of a play of the parity game.
An \emph{even \icw} is a sequence of (not necessarily consecutive) positions
of $\rho$ 
\begin{equation*}
p_1, p_2, p_3, \dots, p_{\ell_i}
\end{equation*}
of length exactly $\ell_i$,
and an \emph{odd \icw}
is a sequence of (not necessarily consecutive) positions
of $\rho$ 
\begin{equation*}
p_0, p_1, p_2,\dots, p_{\ell_i} 
\end{equation*}
of length exactly $\ell_i+1$,
that satisfy the following properties:
\begin{itemize}
\item \textbf{Position:} each $p_j$ specifies a position in the play prefix $\rho$,
so each $p_j$ is an positive integer that satisfies $1 \le p_j \le m$.
\item \textbf{Order:} the positions are ordered. So we have $p_j < p_{j+1}$ for
all $j < \ell_i$.

\item {\bf Evenness:} 
all positions but the final one are even. Formally, for all $j < \ell_i$ the
colour $\phi(v_{p_{j}})$ of the vertex in position $p_{j}$ is even.

For position $p_{\ell_i}$, its colour $\phi(v_{p_{\ell_i}}) = i$. Then, the colour of that position is even for even \icw, and odd for odd \icw.

Note that this entails that an \icw\ has $\ell_i$ initial even positions that define an even chain of length $\ell_i$.

\item {\bf Inner domination:}  
the colour of every vertex between $p_j$ and $p_{j+1}$ is dominated by the colour
of $p_j$ \emph{or} the colour of $p_{j+1}$.
Formally, 
for all $j < \ell_i$, the colour of every vertex in the subsequence
$v_{p_{j}},v_{p_{j}+1},\ldots,v_{p_{j+1}}$ is less than or equal to 
$\max\big\{\phi(v_{p_{j}}),\phi(v_{p_{j+1}})\big\}$.

\item {\bf Outer domination:} 
the colour of the vertex $v_{p_{\ell_i}}$ in position $p_{\ell_i}$ is $i$, i.e.\ $i = \phi(v_{p_{\ell_i}})$.
Moreover, $i$ is greater than or equal to the colour of every vertex
that appears after position
$p_{\ell_i}$ in $\rho$.
Formally, for all $k$ in the range $p_{\ell_i} \le k \le m$, we
have that $\phi(v_{k}) \le i$.
\end{itemize}

\paragraph{\bf Colour witnesses}
Like a concise witness, a \emph{colour witness} is a sequence 
\begin{equation*}
b_k, b_{k-1}, \dots, b_1, b_0,
\end{equation*}
of length%
\footnote{$k = \lfloor \log_2(e)\rfloor$ again suffices, where $e$ is the number of vertices with an even colour} $k+1$, such that each element $b_i \in \cu$, and
that satisfies the following properties.

\begin{itemize}

\item \textbf{Properties of the sequence:}

defining \textbf{$i$-positions} as the positions in $\mathbf b$ that have value $i$,

$\pos(i,\mathbf b) = \{j \leq k \mid b_j = i \}$ for every $i \in C^-$, the sequence has to satisfy the following constraints:

\begin{itemize}
\item \textbf{order:} for $i>j$, we have that $b_i \geq b_j$ or $\_\in \{b_i,b_j\}$ holds; and

\item \textbf{conciseness:} for all odd $i \in C^-$, $\big|\pos(i,\mathbf b)\big| \leq 1$ and $b_0 \neq i$ hold.
\end{itemize} 

\item \textbf{Witnessing:}
\begin{itemize} 
\item \textbf{ordered witnesses:} for $i>j$ with $\pos(i,\mathbf b)\neq \emptyset$ and $\pos(j,\mathbf b)\neq \emptyset$, the $j$-witness starts after the
$i$-witness ends.
That is $p_{i, \ell_i} < p_{j, 1}$ if $j$ is even and $p_{i, \ell_i} < p_{j, 0}$ if $j$ is odd.
\item using the following definitions,
\begin{itemize}
\item \textbf{next odd colour:} $\odd(i,\mathbf b) = \inf\{j > i \mid j \mbox{ odd and }\pos(j,\mathbf b)\neq \emptyset\}$ defines the next higher odd colour than $i$ that occurs in the colour witness (note that $\odd(i,\mathbf b)= \infty$ if no such colour exists), 

\item \textbf{unblocked colours:} $\unblk(i,\mathbf b) = \{ j < \odd(i,\mathbf b) \mid j \geq i \mbox{ and }\pos(i,\mathbf b) \neq \emptyset\}$ is the set of all colours that are at least $i$, but strictly smaller than $\odd(i,\mathbf{b}$, and

\item \textbf{unblocked positions:} $\ubp(i,\mathbf b) = \bigcup\limits_{j \in \unblk(i,\mathbf b)}\pos(i,\mathbf b)$ is the set of positions labelled by an unblocked colour,
\end{itemize}
\medskip
we have that $\sum\limits_{\unblk(i,\mathbf b)} \ell_i \geq \sum\limits_{j \in \ubp(i,\mathbf b)}2^i$ holds for all $i \in C^-$.
\end{itemize}
\end{itemize} 

It should be noted that neither the \icw-s associated with each colour, nor the value of the $\ell_i$ are stored in a colour witness.
However, in order for a sequence to be a colour witness for an initial sequence of a run, the
corresponding \icw-s must \emph{exist}.

\subsection{Updating Colour Witnesses}

We now show how forward colour witnesses can be constructed incrementally by processing
the play one vertex at a time.
Throughout this subsection, we will suppose that
we have a play $\rho = v_0, v_1, \dots, v_{m}$, and a new vertex $v_{m+1}$ that
we would like to append to $\rho$ to create $\rho'$. 
We will use $d = \phi(v_{m+1})$ to denote the colour of this new vertex. 
We will suppose that $\mathbf b = b_k, b_{k-1}, \ldots, b_1,b_0$ is a colour witness for $\rho$,
and has \icw-s with individual lengths $\ell_i$.
We will construct a witness $\mathbf c = c_k, c_{k-1}, \ldots, c_1,c_0$ for $\rho'$ and discuss how its inferred \icw-s look like.

We present four lemmas that allow us to perform this task.

\begin{lemma}
\label{lem:newup.stale}
Suppose that $d\in C^-$ is odd and, for all $j \leq k$, either $b_j = \_$ or $b_j > d$.
If we set $c_i = b_i$ for all $i \leq k$, then $\mathbf c$ is a colour witness for $\rho'$.
\end{lemma}

\begin{proof}
	Since $d < b_j$ for all $j$, the outer domination of every $e$-colour witness
	implied by $\mathbf b$ is not changed.
	Moreover, no other property of any $e$-colour witness
	is changed by the inclusion of $v_{m+1}$ in the initial sequence, so by setting $\mathbf c = \mathbf b$
	we obtain a colour witness for $\rho'$.
	\qed
\end{proof}

Note that the proof of Lemma \ref{lem:newup.stale} does not use that $d$ is odd and holds similarly when $d$ is even; however, in that case Lemma \ref{lem:newup.overflow} provides a better update for the colour witness.

\begin{lemma}
\label{lem:oddup.local}
Suppose that $d\in C^-$ is odd, and there exists an index $j$ such that $b_j \neq \_$, $d \geq b_j$, and, for all $i>j$, either $b_i = \_$ or  $b_i > d$ hold.
Then setting:
\begin{itemize}
    \item $c_i = b_i$ for all $i > j$,
    \item $c_j = d$ if $j\neq 0$ and $c_j =\_$ if $j=0$, and
    \item $c_i = \_$ for all $i < j$
\end{itemize}
yields a colour witness for $\rho'$.
\end{lemma}

\begin{proof}
For all $e > d$, the $e$-colour witness (if any) implied by $\mathbf b$ can be kept: the outer domination of every such $e$-colour witness implied by $\mathbf b$ is not changed.
Moreover, no other property of any such $e$-colour witness is changed by the inclusion of $v_{m+1}$ in the initial sequence.

For the $b_j$-colour witness, we either update the last vertex to $m+1$ (if $b_j$ is odd) or append $m+1$ to the it (if $b_j$ is even).
In both cases, the inner domination rules are valid (due to the inner and outer domination rules for the $b_j$-colour witness) and the outer domination rule holds trivially.
Moreover, the side constraints for the length carry over from those for $\mathbf b$ (when $b_j$ is odd), for $\ell_d$, by adding one to the length constraint while also appending one state (when $b_j$ is even).

So, $\mathbf c$ is a colour witness for $\rho'$.
\qed
\end{proof}

\begin{lemma}
\label{lem:evenup.local}
Suppose that $d$ is even, there exists a maximal index $j$ such that $b_j<d$, and $b_j$ is odd. Then setting:
\begin{itemize}
    \item for all $i \geq j$, $c_i = d$ if $b_i < d$ and $c_i = b_i$ otherwise,
    \item for all $j > i \geq 1$, $c_i = \_$, and
    \item $c_0 = d$.
\end{itemize}
yields a colour witness for $\rho'$.
%Then $\mathbf c$ is a colour witness for $\rho'$.
\end{lemma}

\begin{proof}
We simply append all \icw-s that exist for $\mathbf b$ in the interval $i \in \{b_j,\ldots,d\}$.
We append them in the given order (from the largest $i$ to the the lowest, $b_j$), and then replace the last index (which is from the $b_j$-covit) by $m+1$.

The inner domination rules are valid (due to the inner and outer domination rules for the \icw-s involved, and by $b_j < d$. The outer domination rule trivially holds.

The only new rule to be considered is the rule on the joint length of the \icw-s in $\ubp(d,\mathbf c)$, but this is the same length (as only the last element is changed) and the same constraint as for the sum of the length of the \icw-s in $\ubp(b_j, \mathbf b)$. 
\qed
\end{proof}

\begin{lemma}
\label{lem:newup.overflow}
Suppose that $d$ is even and there is no index $j'$ such that $b_{j'}<d$ and $b_{j'}$ is odd.
Let $j$ be the maximal index (which might be $0$) such that:
\begin{itemize}
    \item for all $i > j$, $b_i$ is even, $b_i=\_$, or $b_i>d$;
    \item either $b_j = \_$, or $b_j > d$ and $b_j$ is odd; and
    \item for all $i < j$, $b_i$ is even.
\end{itemize}
If we set:
\begin{itemize}
    \item $c_i = b_i$ for all $i > j$ with $b_i>d$ or $b_i=\_$
    \item $c_i = d$ for all $i>j$ with $b_i \leq d$ (and thus even),
    \item $c_j = d$, and
    \item for all $i < j$, $b_j = \_$,
\end{itemize}
then $\mathbf{c}$ is a colour witness for $\rho'$.
\end{lemma}

\begin{proof}
We simply append all \icw-s that exist for $\mathbf b$ in the interval $i \in \{2,\ldots,d\}$.
We append them in the given order (from the largest $i$ to the the lowest), and then append $m+1$.

The inner domination rules are valid (due to the inner and outer domination rules for the \icw-s involved. The outer domination rule trivially holds.

The only new rule to be considered is the rule on the joint length of the \icw-s in $\ubp(d,\mathbf c)$, but this is one more than the length (as only the last element is appended) and the same constraint as for the sum of the length of the \icw-s in $\ubp(2, \mathbf b)$.
\qed
\end{proof}

Again, if $d = \max\{C\}$ and odd, then the raw update of the colour witness is $\_,\ldots,\_$, which is a colour witness for every play prefix.
\medskip

When we want to update a witness upon
scanning another state $v_{m+1}$ with colour $d = \phi(v_{m+1})$, we can apply the update rule from one of the Lemmas \ref{lem:newup.stale} through \ref{lem:newup.overflow}.

For a given witness $\mathbf b$ and a vertex $v_{m+1}$, we denote with
\begin{itemize}
 \item \textbf{Raw update:} $\ru_+(\mathbf b,d)$ the raw update of the witness to $\mathbf c$, as
obtained by the update rules described above.
 \item \textbf{Update:} $\up_+(\mathbf b,d)$ is either $\ru_+(\mathbf b,d)$ if
$\val\big(\ru(\mathbf b,d)\big) \leq e$ (where $e$ is the number of
vertices with even colour), or $\up_+(\mathbf b,v_{m+1}) = \won$ otherwise.

In particular, $\up_+(\won,d) = \won$ for all $d \in C$.

\item \textbf{Antagonistic update:} $\au_+(\mathbf b,v) = {\min}_{\sqsubseteq}\big\{\up_+(\mathbf c,v) \mid \mathbf b {\sqsubseteq} \mathbf c \in \mathbb C \big\}$.
\end{itemize}

We first observe that $\up_+$ is indeed `faster' than $\up$ in that it always leads to a better (w.r.t.\ $\sqsubseteq$) state:
\begin{lemma}
For all $\mathbf b \in \mathbb C$ and all $d\in C$, $\up_+(\mathbf b,d) \sqsupseteq \up(\mathbf b,d)$.
\qed
\end{lemma}

This is easy to check by the raw update rules, and it entails:

\begin{corollary}
For all $\mathbf b \in \mathbb C$ and all $d\in C$, $\au_+(\mathbf b,d) \sqsupseteq \au(\mathbf b,d)$.
\qed
\end{corollary}

\begin{theorem}
\label{theo:fast}
The following three claims are equivalent for both, forward and backward witnesses:
\begin{enumerate}
    \item player \emph{even} has a strategy to win the parity game,
    \item player \emph{even} has a strategy to win the fast basic update game (using $\up_+$), and
    \item player \emph{even} has a strategy to win the fast antagonistic update game  (using $\au_+$).
\end{enumerate}
\end{theorem}

\begin{proof}
(1) implies (3): By Theorem \ref{theo:concise.antagonistic}, that player \emph{even} wins the parity game entails that he wins the concise antagonistic update game.
As $\au_+$ provides (not necessarily strictly) better updates (w.r.t.\ $\sqsubseteq$) than $\au$, and by the antagonistic update being monotone by definition, this entails (3).

(3) implies (2): as $\up_+$ provides (not necessarily strictly) better updates (w.r.t.\ $\sqsubseteq$) than $\au_+$ and $\au_+$ is monotone, when $\au^+$ produces a winning sequence, so does $\up^+$.

(2) implies (1):
when $\up^+$ produces a win if, and only if, $\ru^+$ produces a colour witness with value $>e$, which according to Lemmas \ref{lem:newup.stale} through \ref{lem:newup.overflow} entails that it has an
even chain whose length is strictly greater than $e$.
The play $\rho$ must, at that point, contain a cycle,
since there must be a vertex with even colour that has been visited twice.
Moreover, the largest priority on this cycle must be even, so this is a winning
cycle for player \emph{even}.
\qed
\end{proof}

\begin{corollary}
For a parity game with $e$ states of even colour, colours $C$, $k=\lfloor \log_2(e) \rfloor$ and $\mathbb C$ the space for concise witnesses of value $\leq e$, length $k+1$ and colours $C$, both
$\mathcal U=(\mathbb C;C;\_,\ldots,\_;\up_+;\won)$ and $\mathcal A=(\mathbb C;C;\_,\ldots,\_;\au_+;\won)$ are separating automata.
\qed
\end{corollary}

To give an intuition to their states, for $\mathcal U$ being in a state $\mathbf b \neq \won$ means that $\mathbf b$ is a colour witness for the play prefix, while $\won$ means that the play prefix contains an even chain of length $>e$, and thus an even cycle.

For $\mathcal A$ being in a state $\mathbf b$ means that there is a state $\mathbf c \sqsupseteq \mathbf b$ with this property.

\subsection{Faster Conversion}
While the method will speed up $\mathcal A=(\mathbb C;C;\_,\ldots,\_;\au_+;\won)$ a little, the difference is easier to see when using $\mathcal U=(\mathbb C;C;\_,\ldots,\_;\up_+;\won)$.

The classic QP algorithms \cite{CJKLS17,JL17,FJKSSW19} have very simple pathological examples.
For example, \cite{JL17} would traverse the complete statespace for a state colour $2$ and a selfloop (or for a state with colour $1$ and a selfloop, depending on whose player's side the algorithm takes).
Similarly, \cite{CJKLS17} would traverse very large parts of its statespace when fed with only even colours.

Using our update rules for colour witnesses, a loop with even colours will always lead to acceptance within $e+1$ steps.

It is possible to make this a bit more robust against the occurrence of odd priorities that are then immediately followed by higher even priorities by returning to $\mathbb W$ as a statespace%
\footnote{When using $\mathbb W$, there would need to be some care taken with odd $o$-witnesses:
while the rules for overwriting lower numbers are as expected, the point to bear in mind that the treatment of odd colours that already occur in a witness (covered by  Lemma \ref{lem:oddup.local}) generalise to `if there is already a lowest position $j$ with $b_j=o$, then just replace all $b_i$ with $i<j$ by $\_$.
This is because, while two even colour witnesses can be merged, two odd colour witnesses cannot, and there would be no means to mark them as different colour witnesses.}.

\section{Statespace}
\label{sec:complex}
In this section, we compare the size of the statespace with both the statespace from the construction of Jurdzinski and Lasic~\cite{JL17}--which comes with the best current bounds-- and the original stataspace from Calude et al.~\cite{CJKLS17}.

We then discuss the effect of the four improvements over the original approach from Calude et al.~\cite{CJKLS17}:
\begin{enumerate}
    \item the restriction of the number of occurrences of odd colours in a witness to once,
    \item not using any colour that is higher than any even colour;
    \item not allowing for odd colours in the rightmost position (i.e. $b_0$);
    \item the removal of the colour $1$; and
    \item moving from length to value restriction.
\end{enumerate}

The first of these improvements is, individually, the most powerful one. Three of the other improvements, (2) -- (4), have already been discussed in this form in \cite{FJKSSW19}.

We will show in Subsection \ref{ssec:equal} that applying \emph{only} improvements (1) and (2) leads to a statespace of exactly the same size as that of Jurdzinski and Lasic~\cite{JL17}.

Consequently, the further improvements, (3) -- (5), lead to a strictly smaller statespace.
The improvement from (3) alone almost halves the statespace, while (4) alone has only a small effect. The effect of rule (5) varies greatly: it is strongest when the bound on the length of an even chain is a power of $2$ ($2^p$ for some $p \in \mathbb N$), where it leads to halving the statespace, and vanishes if it is one less ($2^p-1$ for some $p \in \mathbb N$).

After briefly visiting the statespace from \cite{FJKSSW19}, we then turn to an experimental comparison of the three statespaces of interest, confirming the quantification of the advantage we have obtained over \cite{JL17}.

In this section, we use $\cnt^{size}_{alg}$ for counting the number of state minus one, estimating the number of states except for the winning state (`$\won$'), which all progress measures under consideration have.
The superscript \emph{size} can be $\ell$, saying that only the length of the data structure (or: the $\lceil \log_2 (e+1) \rceil$ for the maximal length $e$ of an even chain) is taken into account; $v$ if the value of witness is taken into account (or: the maximal length $e$ of an even chain) is taken into account, and $\ell, v$ if both are used.
The subscript is either $JL$ when counting the concise progress measures from \cite{JL17}, $O$ when considering the original approach from \cite{CJKLS17}, $'1,2'$ when adding improvements (1) and (2), or blank when considering either improvements (1) through (4) or all improvements. In a closing comparison with the statespace of \cite{FJKSSW19}, we use the subscript $JKSSW$.

\subsection{Concise Progress Measures \cite{JL17}}
While we will not describe the algorithm, but the data structure, which holds a winning state besides the states we describe.
For each even priority, there is a (possibly empty) word over two symbols, say $+$ and $-$, such  that the words concatenated have length at most
$\ell = \lceil \log_2 (e + 1) \rceil$, where $e$ is the number of states with an even priority. That is, $\ell$ is the length of the witness and colour witness from the previous section ($\ell = k+1$).

For \cite{JL17} (i.e. \emph{alg} = $JL$), with $c$ priorities $\{1,\ldots, c\}$ and $n$ states with even priority (not counting the winning state) we have the following counts:

\begin{itemize}
    \item Induction basis, length: we start with the case in which we bound the sum of the lengths of $+$ and $-$ by $0$ or $1$.
    
    When we bound the sum of the lengths by $0$, then there is only one sequence:
    
    $$\cnt^\ell_{JL}(c,0)=1\ ,$$
    
    and when we bound it by $1$, then we get:
    $$\cnt^\ell_{JL}(2c,1) = \cnt^\ell_{JL}(2c+1,1) = 2c+1\ , $$
    
    as there are $c$ positions in which a sequence of length $1$ can occur (one for each even priority in $\{1,\ldots, 2c+1\}$ or $\{1,\ldots, 2c\}$, respectively), and there is one for the case in which all sequences have length $0$.
    
    \item Induction basis, colours:
    when there is only one even colour (i.e. $2$), we have:
    
    $$\cnt^\ell_{JL}(3,l) = \cnt^\ell_{JL}(2,l) = 2^{l+1}-1 \ . $$
    
    These are the binary words of length at most $l$.
    
    \item For all other cases, we define inductively:
    $$\cnt^\ell_{JL}(2c+1,l) = \cnt^\ell_{JL}(2c,l) = \cnt^\ell_{JL}(2c-2,l) + 2\cnt^\ell_{JL}(2c,l-1)\ ,$$
    
    where the first summand refers to the case where the leading sequence (which refers to colour $2c$) is empty. In this case, the length of the remaining sequences is still bound by $l$, but the number of even colours has dropped by one.
    The two summands $\cnt^\ell_{JL}(2c,l-1)$ represent the cases, where the sequence assigned to the highest even priority starts with a $+$ and $-$, respectively.
    Cutting off this leading sign leaves $\cnt^\ell_{JL}(2c,l-1)$ different states.
\end{itemize}

To estimate the number of concise progress measures, $\cnt^\ell_{JL}(c,l)+1$,
we put aside the winning state and the function that maps all even priorities to the empty sequence.

For the remaining states, we first fix a positive length $i\leq l$ of the concatenated words, and then the $j \leq \lfloor c/2 \rfloor$ of the even priorities that have a non-empty word assigned to them.

There are $\Big(\begin{array}{c} i - 1  \\ j  - 1 \end{array}\Big)$ assignments of positive lengths to $j$ positions that add up to $i$.
For each distribution of lengths, there are $2^i$ different assignments to words.
Finally there are $\Big(\begin{array}{c}\lfloor c/2 \rfloor   \\ j \end{array}\Big)$ possibilities to assign $j$ of the $\lfloor c/2 \rfloor$ different even priorities.

This provides an overall statespace of

$$2 + \sum_{i=1}^l \sum_{j=1}^{\min\{{i,\lfloor c/2 \rfloor }\}} 2^i 
\cdot  \Big(\begin{array}{c}\lfloor c/2 \rfloor   \\ j \end{array}\Big) \cdot \Big(\begin{array}{c} i - 1  \\ j  - 1 \end{array}\Big) \ .$$

\subsection{Calude et al.~\cite{CJKLS17}}
We now continue with the statespace of the original quasi-polynomial approach of Calude et al.~\cite{CJKLS17}, hence, \emph{alg} = $O$.
The statespace used in \cite{CJKLS17} is slightly larger than $\mathbb W$, as it only uses the length of the witness as restriction and does not exclude odd values for the rightmost position (`$b_0$') in a witness.

For the precise count of this statepace without the winning state `$\won$', we have the following counts:

\begin{itemize}
    \item Induction basis, length:
    sequences of length $1$ (only containing $b_0$) whose colour values are bounded by $c$, can take $c+1$ values in $\{1,\ldots, c\}$ plus $\_$. We therefore have:
    
    $$\cnt^\ell_O(c,1) = c +1 \ .$$ 

    \item For longer tails of sequences, we define inductively:
    
    $$\cnt_O(c,l+1) = \cnt_O(c,l) + \sum_{i=1}^{c}\cnt_O(i,l)\ .$$
    
    The summands refers to the possible values taken by the leftmost postion (`$b_l$') of the tail ($b_l,b_{l-1},\ldots,b_0$).
    The first summand refers to the leading position being $\_$ (`$b_l=\_$').
    This does not restrict the values the rest of the tail may take any further as the highest colour allowed to appear is still $c$.
    The other summands refer to the value $i$ (`$b_l=i$').
    When the value of the leftmost position is $i\leq c$, then the highest colour that may occur in the remaining positions is $i$.
    \end{itemize}

The size of $|\mathbb W^+|$ of the statespace can be given as:

$$2 + \sum_{i=1}^l \Big(\begin{array}{c} l  \\ i \end{array}\Big)\cdot \Big(\begin{array}{c}i + c - 1  \\ i \end{array}\Big) \ .$$

The `$2$' refers to the winning state and the `empty' sequence that consists only of $\_$ symbols (which is more convenient for us to treat separately), and the sum refers to the states represented by non-empty sequences of length $l = \lceil \log_2 (e + 1) \rceil$, where $e$ is the number of states with an even priority. Note that the estimation given in \cite{CJKLS17} is slightly coarser, and their game definition is slightly different from the normal definition of parity games, but the bound can be taken from \cite{FJKSSW19}.

For the estimation, after fixing the positive $i \leq l$ positions with values different to $\_$, there are $\Big(\begin{array}{c}i + c - 1  \\ i \end{array}\Big)$ different valuations when we have $c$ priorities.
For each $i \leq l$, there are $\Big(\begin{array}{c} l  \\ i \end{array}\Big)$ different choices for the $i$ positions containing some number $d \in C$.
This leads to $ \sum_{i=1}^l \Big(\begin{array}{c} l  \\ i \end{array}\Big)\cdot \Big(\begin{array}{c}i + c - 1  \\ i \end{array}\Big)$ different states that contain $i\leq l$ positions that have a value in $C$.

To obtain a better inroad to outline the differences, we first take a closer look at the $\Big(\begin{array}{c}i + c - 1  \\ i \end{array}\Big)$ different valuations that we may have for $c$ priorities when we have fixed the $i \leq l$ positions that are not marked as $\_$.

For these positions, we can look at the cases where there are $j\leq \min(i,r)$ fixed different priorities. For that case, there are $\Big(\begin{array}{c} i - 1  \\ j  - 1 \end{array}\Big)$ many assignments of these $j$ priorities to the $i$ positions.
Moreover, there are $\Big(\begin{array}{c} c  \\ j \end{array}\Big)$ different options to select $j$ of the available $r$ priorities.
Thus, we get

$$\Big(\begin{array}{c}i + c - 1  \\ i \end{array}\Big) = \sum_{j=1}^{\min\{{i,c}\}} \Big(\begin{array}{c} c  \\ j \end{array}\Big) \cdot \Big(\begin{array}{c} i - 1  \\ j  - 1 \end{array}\Big)$$
different combinations for all number of priorities put together, providing the following size:

$$2 + \sum_{i=1}^l \sum_{j=1}^{\min\{{i,c}\}} \Big(\begin{array}{c} l  \\ i \end{array}\Big)\cdot  \Big(\begin{array}{c} c  \\ j \end{array}\Big) \cdot \Big(\begin{array}{c} i - 1  \\ j  - 1 \end{array}\Big) \ .$$

\subsection{Improvements}
We now discuss the differences obtained when moving from $\mathbb W^+$ to $\mathbb C$ by looking at the effect of the three optimisations we have introduced. These are:
\begin{enumerate}
    \item the restriction of the number of occurrences of odd colours in a witness to once,
    \item not using any colour that is higher than any even colour;
    \item not allowing for odd colours in the rightmost position (`$b_0$');
    \item the removal of the colour $1$; and
    \item moving from length to value restriction.
\end{enumerate}

\subsubsection{(1) and (2) Restricted Occurrence of Odd Colours.}
\label{ssec:equal}
Restricting the occurrence of odd colours to once, together with the optimization of not using any colour that is higher than any even colour, leads to a situation where the highest colour allowed in any position is even.
To see this, we observe that the banning of a potential odd colour higher then any even colour guarantees this initially, where the highest colour allowed is the highest even colour.

When an odd colour $o$ is used in the witness (`$b_l = o$'), then the highest colour allowed to its right is $o-1$, whereas when an even colour $e$ is used in the witness (`$b_l = e$'), then the highest colour allowed to its right is $e$.

We therefore only have to define our improved counting function for even colours:

\begin{itemize}
    \item Induction basis, length:
    apart from using only even bounds, the base case remains the same:
    
    $$\cnt^\ell_{1,2}(2c,1) = 2c +1 \ .$$ 

    \item For longer tails of sequences, we define inductively:
    
    $$\cnt^\ell_{1,2}(2c,l+1) = 1 + 2\sum_{i=1}^{c}\cnt^\ell_{1,2}(2i,l)\ .$$
    
    The summands refers to the possible values taken by the leftmost postion (`$b_l$') of the tail ($b_l,b_{l-1},\ldots,b_0$).
    
    The first summand refers to the leading position being $1$ (`$b_l=1$'). If this is the case, then all entries to its right must be \emph{strictly smaller} than $1$ (which is not possible) or $\_$ -- consequently, they must all be $\_$, which just leaves one such tail.
    
    The other summands refer to the leading position taking the value $2i$ or $2i + 1$ when $i< c$ (`$b_l=2i$' or `$b_l=2i+1$'), in either case, the maximal value of the colours occurring in the remaining tale is $2i$.
    
    The final two summands (for $i=c$) refer to the leading position taking the value $2c$ or $\_$ (`$b_l=2c$' or `$b_l=\_$').
    In both cases, the maximal value of the colours occurring in the remaining tale is $2c$.
    \end{itemize}

While the representation is different, it is easy to see that $\cnt^\ell_{1,2}(2c,l) = \cnt^\ell_{JL}(2c,l)  $ holds.

To see this, we first observe that $\cnt^\ell_{JL}(2,l+1) = 1 + 2\cnt^\ell_{JL}(2,l)$ holds, and then by induction over $c$ that:

$$\cnt^\ell_{JL}(2,l+1) = 1 + 2\sum_{i=1}^c\cnt^\ell_{JL}(2i,l)\ .$$ 

Given that we also have $\cnt^\ell_{1,2}(2c,1) = \cnt^\ell_{JL}(2c,1)$, we get the claim, because $\cnt^\ell_{JL}(2c,l)$ cannot be derived in the same way as $\cnt^\ell_{1,2}(2c,l)$.

\subsubsection{(3) and (4) Removing odd colours from the rightmost position and $1$s.}
Removing odd colours from the rightmost positions only changes the base case, while banning $1$ from the other positions merely removes the `$1+$' part from the inductive definition. This leaves:

\begin{itemize}
    \item Induction basis, length:
     
    $$\cnt^\ell(2c,1) = c +1 \ .$$ 

    \item For longer tails of sequences, we define inductively:
    
    $$\cnt^\ell(2c,l+1) = 2\sum_{i=1}^{c}\cnt^\ell(2i,l)\ .$$
\end{itemize}

When evaluating the term $\cnt^\ell$, the reduction from $2c + 1$ to $c+1$ is halving the value (rounded up) at the leaf of each call tree, which provides more than the removing of `$=1$' in each node of the call tree. Together, they \textbf{broadly halve the value}.

\subsubsection{(5) Taking the value into account.}
We start with using both the length and the value, and then remove the length in next step to get a more concise representation, but we note that, for a given length $l$, the value $v$ allowed always satisfies $v< 2^l$.

First, we get another induction base, one by value:

\begin{itemize}
    \item Induction basis, value:
     
    $$\cnt^{\ell,v}(2c,l,0) = 1 \ .$$ 

Regardless of the remaining length, if the \emph{value} of the tail is bounded by (and thus needs to be) $0$, then it can only consist of $\_$ signs.

    \item Induction basis, length:
     
    $$\cnt^{\ell,v}(2c,1,1) = c +1 \ .$$ 

    \item For longer tails of sequences and positive values, we distinguish a number of cases. The first case is that $v < 2^{l}$. Then we have
    
    $$\cnt^{\ell,v}(2c,l+1,v) = \cnt^{\ell,v}(2c,l,v)\ .$$
    
    This is simply because filling the position $l+1$ with any number, even or odd, would exceed the value budget.
    
    This leaves the case $v \geq 2^{l}$, that is:
\[
\begin{array}{ccl}
\cnt^{\ell,v}(2c,l+1,v) =& &
    \sum_{i=1}^{c}\cnt^{\ell,v}(2i,l,v-2^l)\\[5pt]
    &+&  \sum_{i=1}^{c}\cnt^{\ell,v}(2i,l,2^l-1) \ .
    \end{array}
\]
This is because, when filling position $l$ with an even number, it takes $2^l$ from the budget of the value, leaving an remaining budget of $v-2^l$.

When filling this position with an odd number, while the value would be increased by $2^l$, this is within the value budget.
Moreover, if this position is still relevant to the value, then the positions to its right no longer add to the value of the sequence, as the leftmost odd position would be the last to be considered.

We therefore set the value for the remaining tail to the right to be the maximal value that can be obtained by this tail, which is $2^l-1$; this is a rendering of saying that for the tail the values are not constrained.
\end{itemize}

The effect of adding the value can vary greatly. It is larger when the number of positions with even colour is a power of $2$, say $2^l$, and it has no effect at all if it is $2^{2}-1$.
In the former case, if the initial position is even, then all other positions need to be $\_$. Generally, we have

\[
\begin{array}{rl}
\cnt^{\ell,v}(2c,l,2^l-1) =& \cnt^{\ell}(2c,l) \qquad \mbox{ and}\\[5pt]
\cnt^{\ell,v}(2c,l,2^{l-1}) =& \cnt^{\ell}(2c,l)/2 + c 
    \end{array}
\]
for all $l > 1$.

\begin{quote}
Taking the value into account therefore broadly \textbf{halves the statespace when $e$ is a power of $2$}, and \textbf{has no effect when $e$ is a predecessor of a power of $2$}, and falls from $2^{l-1}$ to $2^l-1$ for all $l>1$.     
\end{quote} 
\medskip

Looking at the definition of $\cnt^{\ell,v}$, it is easy to see that an explicit reference to the length can be replaced by a reference to the next relevant length, $\lfloor \log_2 v \rfloor$. This provides:

\[
\begin{array}{ccl}
\cnt^{v}(2c,0) =& &1\ , \\[5pt]
\cnt^{v}(2c,1) =& &c+1\ ,\mbox{ and} \\[5pt]
\cnt^{v}(2c,v) =& &
    \sum_{i=1}^{c}\cnt^{v}(2i,v-2^{\lfloor \log_2 v \rfloor})\\[5pt]
    &+&  \sum_{i=1}^{c}\cnt^{v}(2i,2^{\lfloor \log_2 v \rfloor}-1) \mbox{ otherwise.}
    \end{array}
\]
    
\subsection{Comparison with the statespace of \texorpdfstring{\cite{FJKSSW19}}{}}
While improvement (1) is the most powerful of the optimizations, the improvements (2) -- (4) were present in \cite{FJKSSW19}, where the algorithm makes use of a value function, namely $\val'(\mathbf b) = \sum\limits_{i\in \even(\mathbf b)} 2^i$.
It is therefore interesting to provide a count function for \cite{FJKSSW19}. We use the subscript $JKSSW$, and only use the count that uses both length and value.

We get the following state counts:

\[
\begin{array}{rcl}
\cnt^{\ell,v}_{JKSSW}(c,1,0) =& & 1 \\[5pt] 
\cnt^{\ell,v}_{JKSSW}(c,1,1) =& & \lfloor c/2 \rfloor + 1 \\[5pt] 
\mbox{if }v < 2^l:
\cnt^{\ell,v}_{JKSSW}(c,l+1,v) =& & \cnt^{\ell,v}_{JKSSW}(c,l,v) \\[5pt] 
 & + &
 \sum_{i=2}^{\lceil c/2 \rceil}\cnt^{\ell,v}_{JKSSW}(2i-1,l,v)\\[5pt] 
\mbox{if }v \geq 2^l:
\cnt^{\ell,v}_{JKSSW}(c,l+1,v) =& &
    \cnt^{\ell,v}_{JKSSW}(c,l,2^l-1)\\[5pt] 
    &+&  \sum_{i=1}^{\lfloor c/2 \rfloor}\cnt^{\ell,v}_{JKSSW}(2i,l,v-2^l) \\[5pt] 
    &+&  \sum_{i=2}^{\lceil c/2 \rceil}\cnt^{\ell,v}_{JKSSW}(2i-1,l,2^l-1) \ .
    \end{array}
\]

To explain the difference to $\cnt^{\ell,v}$, one major difference is that the highest colour allowed in a position can be odd. The other is that positions with odd colour do not contribute to the weight, which allows for adding positions with odd colour the remaining budget is lower than $2^l$.

Thus, the call tree for the calculation of $\cnt^{\ell,v}_{JKSSW}$ has $\lceil c/2 \rceil$ successors where $v<2^l$ while the call tree for $\cnt^{\ell,v}$ has just one.
For $v \geq 2^l$, the call tree has the same number of successors (for even $c$) or just one additional successor (for odd $c$), but the parameter falls slower.

\subsection{Statespace comparison}
\label{sec:compare}
In this subsection we provide a graphical representation of the statespace size for the three algorithms: Calude et al.\cite{CJKLS17}, Jurdzinski and Lasic \cite{JL17}, and the improvement described in this article.
The size of the statespace on which an algorithm works does not represent how good the algorithm performs in practice. Indeed, in the context of parity games, there are quasi-polynomial time algorithms that behaves like brute-force approaches. Therefore, they always require quasi-polynomial many steps to compute the solution, while most of the exponential time algorithms, instead, almost visit a polynomial fraction of their statespace. The first improvement we described does not affect the performance of the algorithm, since both the original and the improved algorithm requires the same number of steps to solve a game, but the latter works on a reduced statespace.
To measure how big is the cut we consider in Figure \ref{fig:stsp} games with a fixed number of colours and in Figure \ref{fig:perB} games with a linear number of colours in the size of the game. The games of Figure \ref{fig:stsp} range from $2^3$ to $2^{15}$ positions $n$. Therefore, the length of the measure, that is logarithmic in $n$, constantly increases, while the colours are fixed to value 10. As a consequence, the ratio of colours with respect to $n$ range from $80\%$ to $0.02\%$.
As expected, the cut with the original algorithm significantly increases for games that are not dense of colours as the lines tend to diverge on a logarithmic scale. The ratio between Jurdzinski and Lasic approach ($JL$) and the new improvement, instead, converge to a cut of $73\%$ of the statespace.
The games of Figure \ref{fig:perB}, instead range from $2^8$ to $2^9$ positions $n$, so that the length of the measure is fixed, while the number colours constantly grows from 26 to 50. As a consequence, we have that the ratio of colours with respect to $n$ is fixed to $10\%$. In this case, the scale are linear and, even if the improved statespace is always smaller than the other two, the cut tends to shrink.
\begin{figure}
  \begin{tikzpicture}[scale = 1.35, every node/.style = {scale = 0.9}]
      \begin{axis}
        [
          ymin = 0.2, ymax = 3460000, xmin = 8, xmax = 32768,
          xmode=log,
          ymode=log,
          log basis x={2},
          extra y ticks = 5,
          extra y tick labels = {10},
          x axis line style = -, y axis line style = -,
          ymajorgrids = true,
          xlabel = Game positions, ylabel = Number of states / $10^{3}$,
          legend pos = north west,
          legend entries = {Old, JL, New}
        ]
    
        \addplot [green!50!black, solid, line width = 1pt, mark = o, mark
          options = solid, mark size = 2] table [x index = 1, y index = 4]
          {fix10exp3};
    
        \addplot [black, solid, line width = 0.75pt, mark = star, mark
          options = solid, mark size = 2.5] table [x index = 1, y index = 5]
          {fix10exp3};
          
        \addplot [red!50!black, solid, line width = 0.75pt, mark = square, mark
          options = solid, mark size = 2.5] table [x index = 1, y index = 6]
          {fix10exp3};
    
      \end{axis}
  \end{tikzpicture}
\vspace{-3em}
\end{figure}
\begin{table}
\hspace{6em}
    \scalebox{1.50}[1.30]
    {
    \begin{tabular}{|r|c||r|r|r|}
        \hline \multicolumn{2}{|c||}{\ } & \multicolumn{3}{c|}{\ }\\[-0.90em]
        \multicolumn{2}{|c||}{Game data}  & \multicolumn{3}{c|}{Statespace size / $10^{3}$}\\
        \hline
        \multicolumn{1}{|c|}{Nodes} & \multicolumn{1}{c||}{Colours} & \multicolumn{1}{c|}{Old} & \multicolumn{1}{c|}{JL} & \multicolumn{1}{c|}{New}\\
        \hline
        8 & 8 & 2 & 1 & $>$1\\
        16 & 10 & 8 & 5 & 1\\
        32 & 10 & 33 & 18 & 5\\
        64 & 10 & 122 & 61 & 17\\
        128 & 10 & 432 & 187 & 52\\
        256 & 10 & 1462 & 553 & 154\\
        512 & 10 & 4780 & 1579 & 439\\
        1024 & 10 & 15157 & 4374 & 1211\\
        2048 & 10 & 46813 & 11829 & 3261\\
        4096 & 10 & 141264 & 31326 & 8601\\
        8192 & 10 & 417577 & 81461 & 22282\\
        16384 & 10 & 1211700 & 208470 & 56819\\
        32768 & 10 & 3458200 & 525991 & 142884\\
        \hline
    \end{tabular}
    }
\vspace{1em}
\caption{\label{fig:stsp} Size of the statespace for games with a fixed number of colours on a logarithmic scale.}
\end{table}

\begin{figure}
  \begin{tikzpicture}[scale = 1.45, every node/.style = {scale = 0.9}]
      \begin{axis}
        [
          ymin = 190, ymax = 37000, xmin = 260, xmax = 500,
          x axis line style = -, y axis line style = -,
          ymajorgrids = true,
          xlabel = Game positions, ylabel = Number of states / $10^{6}$,
          legend pos = north west,
          legend entries = {Old, JL, New}
        ]

        \addplot [green!50!black, solid, line width = 1pt, mark = o, mark
          options = solid, mark size = 2] table [x index = 0, y index = 3]
          {linear10lin2};

        \addplot [black, solid, line width = 0.75pt, mark = star, mark
          options = solid, mark size = 2.5] table [x index = 0, y index = 4]
          {linear10lin2};
          
        \addplot [red!50!black, solid, line width = 0.75pt, mark = square, mark
          options = solid, mark size = 2.5] table [x index = 0, y index = 5]
          {linear10lin2};

      \end{axis}
    \end{tikzpicture}
\end{figure}

\begin{table}
    \scalebox{1.50}[1.30]
    {
\hspace{6em}
    \begin{tabular}{|c|c||r|r|r|}
        \hline \multicolumn{2}{|c||}{\ } & \multicolumn{3}{c|}{\ }\\[-0.90em]
        \multicolumn{2}{|c||}{Game data}  & \multicolumn{3}{c|}{Statespace size / $10^{6}$}\\
        \hline
        \multicolumn{1}{|c|}{Nodes} & \multicolumn{1}{c||}{Colours} & \multicolumn{1}{c|}{Old} & \multicolumn{1}{c|}{JL} & \multicolumn{1}{c|}{New}\\
        \hline
        260 & 26 & 381 & 190 & 53\\
        280 & 28 & 622 & 318 & 90\\
        300 & 30 & 987 & 518 & 148\\
        320 & 32 & 11531 & 820 & 251\\
        340 & 34 & 2323 & 1271 & 389\\
        360 & 36 & 3456 & 1928 & 608\\
        380 & 38 & 5054 & 2870 & 926\\
        400 & 40 & 7271 & 4201 & 1759\\
        420 & 42 & 10309 & 6053 & 2584\\
        440 & 44 & 14420 & 8596 & 3724\\
        460 & 46 & 19919 & 12047 & 5838\\
        480 & 48 & 27199 & 16675 & 8625\\
        500 & 50 & 36742 & 22818 & 12200\\
        \hline
    \end{tabular}
    }
\vspace{1em}
\caption{\label{fig:perB} Size of the statespace for games with a linear number of colours on a linear scale.}
\end{table}

\section{Discussion}
\label{sec:discuss}
We have introduced three technical improvements over the progress measures used in the original quasipolynomial approach by Calude et al.~\cite{CJKLS17} and its improvements by Fearnley et al.~\cite{FJKSSW19}.
The first two reduce the statespace.

The more powerful of the two is a simple limitation of the occurrences of odd colours in a witness to one.
Where the highest colour is even, this alone reduces the size of the statespace of Calude et al.'s approach to the currently smallest one of Jurdzinski and Lasic \cite{JL17}.
Where the highest colour is odd, we obtain the same by borrowing the simple observation that this highest colour does not need to be used from \cite{FJKSSW19}.

The second new means to reduce the statespace is the only use witnesses that refer to even chains of plausible size, namely those that do not contain more dominating even states then the game has to offer.
A similar idea had been explored in \cite{FJKSSW19}, but our construction is more powerful in reducing the size of the statespace.
The effect of this step ranges from none (where the number of states with even colour is the predecessor of a power of $2$ ($2^\ell-1$ for some $\ell \in \mathbb N$), then rises steeply to a factor of $2$ for a power of $2$ ($2^\ell$), and then slowly falls again, until it vanishes at the next predecessor of a power of $2$.

These improvements work well with the other improvements from \cite{FJKSSW19}, namely not using the colour $1$ and disallowing odd values for the rightmost position (`$b_0$') in a witness.
These improvements broadly halve the statespace, leading to a statespace reduction that broadly oscilates between $2$ and $4$ when compared to the previously leading approach.

The second improvement we have introduced is a re-definition of the semantics of witnesses, moving from the classic \emph{witnesses} to \emph{colour witnesses}.
While it does not lead to a difference in the size of the statespace, it does accelerate its traversal, especially for the `standard' update rule that does not extend to value iteration; in particular it gets rid of the most trivial kind of silly hard examples, such as cliques of states of player odd that all have even colour.

While it clearly accelerates the analysis, it is not as easy as for the statespace reduction to quantify this advantage.

\section{\bf Acknowledgments.}
\includegraphics[height=8pt]{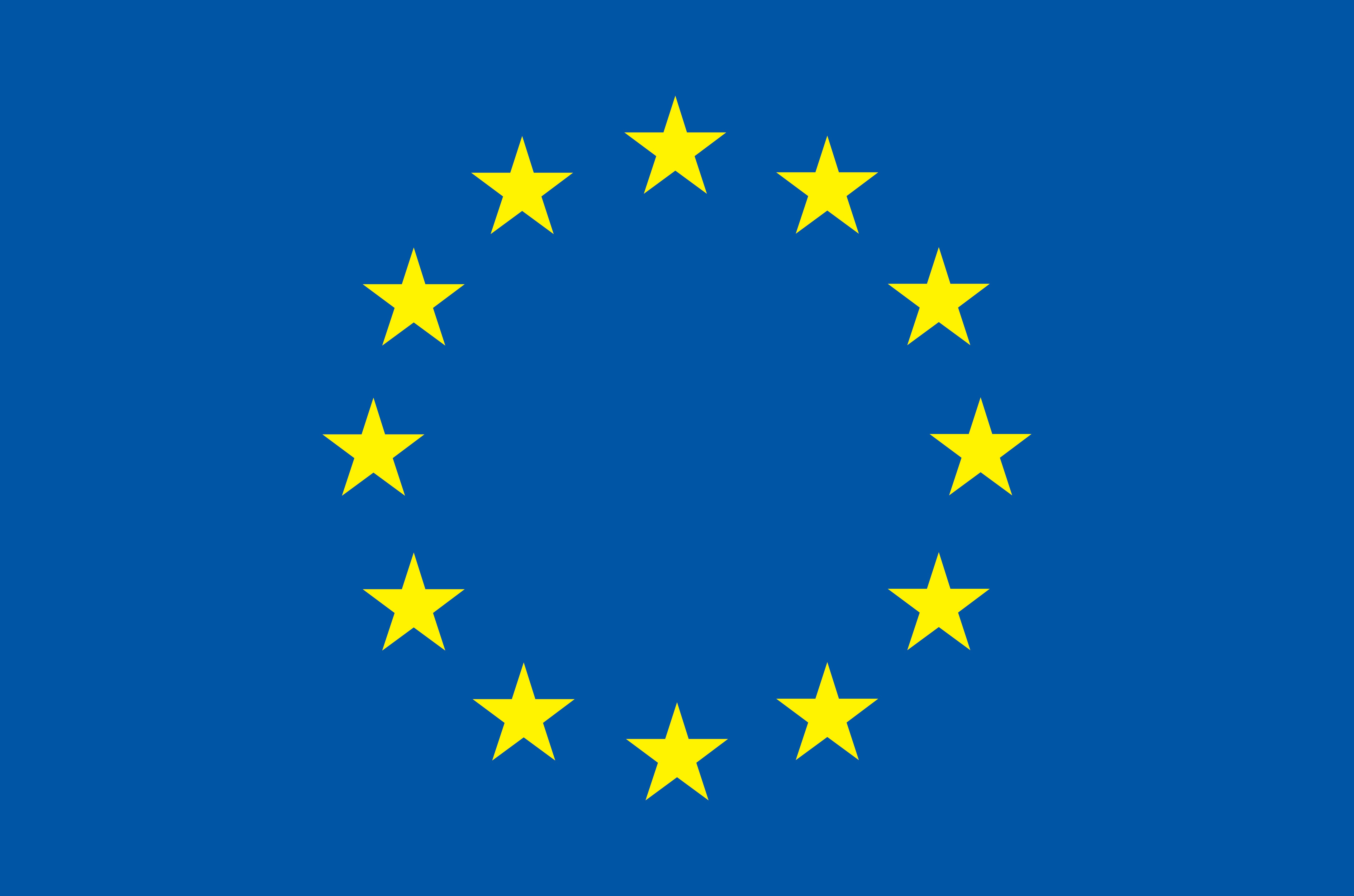} This project has received funding from the European Union’s Horizon 2020 research and innovation programme under the Marie Sk\l odowska-Curie grant agreement No 101032464.

\bibliographystyle{abbrv}
\bibliography{bib}

\end{document}